\documentclass[10pt]{article}

\usepackage[T1]{fontenc}
\usepackage[letterpaper,top=0.72in,bottom=0.78in,left=0.84in,right=0.84in,
  headheight=19pt,headsep=0.18in]{geometry}
\usepackage{amsmath,amssymb,amsthm}
\usepackage{newtxtext,newtxmath}
\usepackage{booktabs}
\usepackage{tabularx}
\usepackage{float}
\usepackage{graphicx}
\usepackage[numbers,sort&compress]{natbib}
\usepackage{xcolor}
\usepackage{microtype}
\usepackage{enumitem}
\usepackage{needspace}
\usepackage{etoolbox}
\AtBeginEnvironment{thebibliography}{\raggedright}
\usepackage[font=small,labelfont=bf,labelsep=period]{caption}
\usepackage{titlesec}
\usepackage{fancyhdr}
\usepackage{tikz}
\usepackage{hyperref}
\usepackage[capitalise,noabbrev]{cleveref}
\newcolumntype{Y}{>{\raggedright\arraybackslash}X}
\newcolumntype{P}[1]{>{\raggedright\arraybackslash}p{#1}}

\definecolor{paperink}{HTML}{17212B}
\definecolor{paperaccent}{HTML}{315C6D}
\definecolor{paperrule}{HTML}{CCD6DC}
\definecolor{orcidgreen}{RGB}{166,206,57}

\hypersetup{
  colorlinks=true,
  linkcolor=paperaccent,
  citecolor=paperaccent,
  urlcolor=paperaccent,
  pdftitle={Who Does Withholding Delay?},
  pdfauthor={Daniel Commey},
  pdfsubject={Open-weight AI release under asymmetric proliferation}
}

\titleformat{\section}
  {\large\sffamily\bfseries\color{paperink}}{\thesection}{0.65em}{}
\titleformat{\subsection}
  {\normalsize\sffamily\bfseries\color{paperink}}{\thesubsection}{0.65em}{}
\titleformat{\subsubsection}
  {\normalsize\sffamily\bfseries\color{paperink}}{\thesubsubsection}{0.65em}{}
\titlespacing*{\section}{0pt}{1.65em}{0.65em}
\titlespacing*{\subsection}{0pt}{1.25em}{0.45em}
\titlespacing*{\subsubsection}{0pt}{1em}{0.35em}

\setlist{nosep,leftmargin=1.7em}
\renewcommand{\headrulewidth}{0.35pt}
\renewcommand{\headrule}{\hbox to\headwidth{\color{paperrule}\leaders\hrule height \headrulewidth\hfill}}
\fancypagestyle{firstpage}{%
  \fancyhf{}
  \fancyfoot[C]{\sffamily\footnotesize\thepage}
  \renewcommand{\headrulewidth}{0pt}
}
\fancypagestyle{plain}{%
  \fancyhf{}
  \fancyhead[L]{\sffamily\footnotesize Who Does Withholding Delay?}
  \fancyhead[R]{\sffamily\footnotesize Daniel Commey}
  \fancyfoot[C]{\sffamily\footnotesize\thepage}
  \renewcommand{\headrulewidth}{0.35pt}
}

\newcommand{\orcidicon}[1]{\href{https://orcid.org/#1}{%
  \begin{tikzpicture}[baseline=-0.66ex]
    \fill[orcidgreen] (0,0) circle (0.62em);
    \node at (0,0) {\color{white}\fontsize{6}{6}\selectfont\textsf{\textbf{iD}}};
  \end{tikzpicture}}}

\renewenvironment{abstract}{%
  \begin{center}\begin{minipage}{0.92\textwidth}
  \small\noindent{\sffamily\bfseries\color{paperink}Abstract}\par\vspace{0.35em}
}{%
  \end{minipage}\end{center}\vspace{0.45em}
}

\newtheorem{proposition}{Proposition}
\newtheorem{corollary}{Corollary}

\theoremstyle{definition}

\theoremstyle{remark}

\newcommand{\E}{\mathbb{E}}

\newcommand{\one}{\mathbf{1}}
\newcommand{\policyC}{\mathsf{C}}
\newcommand{\policyP}{\mathsf{P}}
\newcommand{\policyG}{\mathsf{O_g}}
\newcommand{\policyO}{\mathsf{O_0}}
\newcommand{\dd}{\,\mathrm{d}}

\title{Who Does Withholding Delay?}
\newcommand{\papersubtitle}{A Welfare Model of Open-Weight AI Release\\
Under Asymmetric Proliferation}
\author{Daniel~Commey~\orcidicon{0000-0001-5759-918X}}
\date{}

\makeatletter
\renewcommand{\maketitle}{%
  \begin{center}
    {\fontsize{21}{24}\selectfont\sffamily\bfseries\color{paperink}\@title\par}
    \vspace{0.45em}
    {\fontsize{12.5}{15}\selectfont\sffamily\color{paperaccent}\papersubtitle\par}
    \vspace{1.05em}
    {\normalsize\sffamily\bfseries\@author\par}
    \vspace{0.95em}
    {\color{paperrule}\rule{0.96\textwidth}{0.7pt}}
  \end{center}
  \vspace{0.35em}
}
\makeatother

\begin{document}
\maketitle
\thispagestyle{firstpage}

\begin{abstract}
Withholding a dual-use AI model delays only the actors that lack other routes to
a comparable capability. If sophisticated adversaries obtain substitutes faster
than distributed defenders, restriction can delay defenders more than the
adversaries it targets. We compare controlled access, a defender-first window
followed by public release, safeguarded open weights, and minimally restricted
open weights in a discounted welfare model with actor-specific substitute
acquisition. Under exponential acquisition, restriction gives adversaries a
positive discounted access advantage exactly when they substitute faster than
defenders, and, with equal usefulness, immediate release adds more expected capability at a fixed
horizon to the slower-substituting group. Neither result implies that release is
preferable, because opportunistic misuse, defensive reach, safeguard friction,
and irreversible losses can reverse the ranking. In a linear benchmark, broad
release overtakes control above a unique adversary-substitution threshold
whenever such a threshold exists, and we derive the probability that selected
defenders deploy before both adversary substitution and public release. In a nonlinear implementation, each
of the four policies is optimal somewhere in the parameter space. Three nested
2,048-point designs over thirteen inputs show that policy shares depend strongly
on the chosen parameter bounds. Release records and cybersecurity reports
illustrate the quantities a release review would need to measure and are kept separate from the calibration.
\end{abstract}

\section{Introduction}

Many arguments for withholding a powerful AI model rely on delay: restriction
is meant to slow harmful use and give defenders time to prepare. Whether it does
depends on who is delayed. A state agency, organized criminal group, or
well-funded competitor may be able to steal weights, distill a hosted model,
buy intermediated API access, or use a comparable foreign release.
Resource-constrained defenders often lack those routes. In that setting,
restriction can delay defenders more than the sophisticated actors it targets.

Under restriction, sophisticated adversaries obtain an adequate substitute at a
random time $T_S$ and distributed defenders at $T_D$. We call diffusion through
unequal acquisition channels \emph{asymmetric proliferation}. Under the
exponential benchmark used below, the relevant condition is $\lambda_S>\lambda_D$, a comparison of acquisition rates; in any single realization the defender may still acquire first.
Asymmetric proliferation leaves the welfare comparison open: restriction can still delay some harmful activity and can remain preferable.

Release has two effects. It narrows the capability gap between well-resourced
and constrained actors, and it raises the total capability in circulation.
Welfare depends on which populations gain the most new capability and how they
use it, so the comparison has to be made actor by actor.

The model compares four stylized policies. Controlled access can preserve
provider monitoring and withdrawal, but it may centralize access, expose prompts
or outputs to the provider, or exclude users through price or eligibility. A
defender-first window gives selected recipients an early copy before broader
release; it creates an exclusive defensive head start only if recipients deploy
protection before both adversary substitution and public release. Safeguarded
open weights let a sophisticated user strip the defaults, but the defaults can
still deter opportunistic misuse. Minimally restricted open weights remove
provider recall and reduce default friction for legitimate and malicious users
alike.

\paragraph{Contributions.}
The analytic results are elementary consequences of an exponential-acquisition
benchmark; their value lies in isolating which quantity each argument about
delay relies on. (i)~\cref{prop:inversion} gives the discounted
access imbalance that restriction creates, and \cref{prop:empowerment}
gives the finite-horizon capability each population gains from release.
(ii)~\cref{prop:threshold} gives the adversary-substitution rate
above which broad release beats control in a linear benchmark, and
\cref{prop:externality} shows how defensive productivity moves it.
(iii)~\cref{prop:window} gives the probability that a defender-first
window delivers an exclusive head start. (iv)~A nonlinear four-policy model
optimizes the window length and is evaluated over three nested parameter designs
that vary benefit flows, irreversibility, and defender coverage together with
the substitution rates. (v)~Public release records, cybersecurity evaluations,
and one incident account show which of the model's inputs are observable.

For a release review, the results imply a procedure. For each relevant
population, estimate the time to an adequate substitute under restriction and
the capability that release would add over a stated horizon, and count the users
whom public access newly enables. For a defender-first window, estimate
deployment speed and recipient coverage. For default safeguards, measure misuse
deterrence and friction on legitimate work separately. Access benefits and
persistent losses enter under every policy, so a delay advantage is one term among several.

\section{Related work}

Release runs along a spectrum from staged access through hosted services
and downloadable weights to fully open systems
\cite{solaiman2023gradient}. GPT-2's staged rollout is the canonical
example of using time as a policy instrument \cite{solaiman2019release}, and
structured access describes the hosted end of the spectrum, in which the
developer mediates every use \cite{shevlane2022structured}.
The broader foundation-model literature catalogs benefits from
oversight, innovation, and decentralized deployment alongside misuse and
loss-of-control risks \cite{seger2023opensource}. Kapoor et al.\
\cite{kapoor2024societal} argue that the risk of an open model must be
measured against its substitutes, and identify customizability and weak
monitoring as the two features that generate both benefits and harms. This
paper formalizes one component of that marginal risk: delay measured per actor.

Whether publication helps attackers or defenders more is an older question.
Shevlane and Dafoe \cite{shevlane2020offense} list the factors that decide it for
AI research, among them whether potential attackers would obtain the knowledge
anyway and whether defenders can absorb it and deploy a fix. Our substitution
rates and deployment rate formalize those two factors. Garfinkel and Dafoe
\cite{garfinkel2019offense} analyze how the offense--defense balance changes as
both sides increase investment. In software security, Arora, Telang, and Xu
\cite{arora2008disclosure} derive the socially optimal delay between the
discovery of a vulnerability and its public disclosure, during which the vendor
is expected to patch. The defender-first window has the same structure, with
adversary substitution in place of independent rediscovery and selected-defender
deployment in place of the vendor patch.

Two studies examine the economic incentives for openness. Xu et al.\
\cite{xu2025economics} treat openness as a choice inside an AI value
chain, and Gomes \cite{gomes2026paradox} examines dependency and displaced
risk under concentrated compute. Neither compares actor-specific downstream
security outcomes.

Two game-theoretic treatments are closer. Mladenovic, Courville, and
Gidel \cite{mladenovic2026why} model openness inside an R\&D race, so
their players choose while developing the model; we study release after
development. Landolt et al.\ \cite{landolt2026oracle} study release timing in a
bilevel Stackelberg game in which the defender-first sequence is central. Our
model adds independent adversary substitution, defender coverage, and deployment
as determinants of what that sequence is worth.

The choice of hosted versus downloadable also shapes what can be
enforced. Provider-hosted access supports provider-side monitoring,
termination, and withdrawal, although it may centralize enforcement and expose
prompts or outputs to the provider. Downloadable weights support local audit and
adaptation; their safeguards can be modified directly. Hosted safeguards remain
provider-controlled but may be probed or bypassed
\cite{aisi2025risk,wallace2025worstcase}. Removable defaults keep their value to the extent that they deter casual misuse. The Model Openness Framework
distinguishes many dimensions of openness beyond the parameters themselves
\cite{white2024mof}. We use \emph{open weights} for broadly downloadable
parameters and hold the release of data and training code fixed.

\section{Model}
\label{sec:model}

\subsection{Actors, timing, and policies}

A laboratory $L$ already has a dual-use model. A social planner evaluates
its choice among four release policies:
\begin{align*}
\policyC &: \text{controlled hosted or gated access},\\
\policyP(\tau) &: \text{selected-defender access, then public release at }\tau,\\
\policyG &: \text{immediate open-weight release with default safeguards},\\
\policyO &: \text{immediate minimally restricted open-weight release}.
\end{align*}
Three actor classes use the resulting capabilities. Sophisticated adversaries $S$
have the means to chase substitutes under restriction. Opportunistic
adversaries $O$ become more numerous as public access becomes less costly. A distributed defender population $D$ produces
protection through both direct model use and the heterogeneous
participation of many independent maintainers and organizations.

The laboratory commits to a release policy. Under restriction, sophisticated
adversaries and defenders each choose a costly substitute-acquisition effort;
capabilities then diffuse and downstream harm and benefit are realized. A private
laboratory would maximize $U_L(r)=\pi_r-\chi J_r$, where
$J_r=\int_0^\infty e^{-\rho t}H_r(t)\dd t$ is discounted harm,
$\pi_r$ is release-dependent profit, and $\chi$ is the share of harm it
internalizes. The social planner ranks the same policy outcomes using
\cref{eq:welfare}. The formal analysis concerns the social ranking;
\cref{sec:private} returns to private incentives.

Acquisition payoffs are separable across actors, so each actor's effort can be
computed independently, and the welfare comparisons condition on the resulting
rates. The paper is therefore a decision-theoretic comparison of policy outcomes; a policy-dependent Stackelberg equilibrium and a coupled acquisition contest are left for future work. The effort model microfounds the restricted-access rates only, and post-release conditions enter as parameters.

Under $\policyC$, substitute-acquisition times are independent exponentials,
\begin{equation}
T_S\sim\operatorname{Exp}(\lambda_S),\qquad
T_D\sim\operatorname{Exp}(\lambda_D).
\label{eq:acquisition}
\end{equation}
Each rate summarizes several routes: independent development, theft, leakage,
distillation, access bought through intermediaries, and foreign or competitor
releases. The exponential form is a tractable constant-hazard benchmark. The
baseline immediate-release case sets both artifact-acquisition and
effective-use times to zero; the delay extension below separates the two
events.

$T_i$ is the first moment actor $i$ can use an \emph{adequate}
substitute for the specified capability and decision horizon. A model can
therefore substitute for vulnerability discovery without substituting for
autonomous intrusion, incident reconstruction, or patch deployment.
When actor $i$ has several independent channels $j$ with
$T_{ij}\sim\operatorname{Exp}(\lambda_{ij})$, competing risks give
\begin{equation}
T_i=\min_j T_{ij}\sim\operatorname{Exp}\!\left(\sum_j\lambda_{ij}\right).
\label{eq:competingchannels}
\end{equation}
The aggregate rate is $\lambda_i=\sum_j\lambda_{ij}$. Within a single channel, parallel prerequisites are all satisfied at the
maximum of their completion times, and sequential tasks add their durations;
neither construction generally preserves an exponential effective-access
distribution. The open-release benchmark sets artifact access to zero and carries
any remaining implementation inequality inside $q_i$, friction, and reach. For a
staged or compute-intensive release such as Kimi K3, effective access begins
when the artifact and its complements become available to the actor of
interest.
Because $\E[T_i]=1/\lambda_i$, the baseline defender rate
$\lambda_D=0.70$ corresponds to an expected effective-acquisition time of
about 1.43 years; $\lambda_S=0.95$ corresponds to about 1.05 years.

To separate artifact access from effective use, let the lag $\ell_{iB}\geq0$ denote the
deterministic delay between weight release and effective use by actor
$i\in\{S,D\}$ after broad release $B\in\{\policyG,\policyO\}$. Then
\begin{equation}
X_{i,B}(t)=\one\{t\geq\min(T_i,\ell_{iB})\}.
\label{eq:opendelay}
\end{equation}
The two open tiers share the same delays, so the exercise isolates deployment
from safeguard differences. Setting $\ell_{SB}=\ell_{DB}=0$ recovers the baseline.
Independent substitute acquisition remains available during the delays. In the
defender-first continuation, selected defenders keep reach $n_P$ until broader
deployment, and defenders who acquire substitutes after release use controlled
reach until they deploy the public artifact. Only the strategic capability term
is delayed; broad-benefit and opportunistic-misuse flows keep their baseline
timing.

The acquisition rates admit a simple effort microfoundation. Actor
$i\in\{S,D\}$ chooses effort $e_i\ge 0$, which yields rate
$\lambda_i(e_i)=\bar\lambda_i+k_i e_i$. The actor earns discounted access value
$v_iA_i$ and pays effort cost $c_ie_i^2/2$. Using the exposure derived in
\cref{prop:inversion}, the payoff is
\begin{equation}
u_i(e_i)=v_i\frac{\lambda_i(e_i)}{\rho[\lambda_i(e_i)+\rho]}
-\frac{c_i}{2}e_i^2.
\label{eq:effortpayoff}
\end{equation}
It is strictly concave because
\begin{equation}
u_i''(e_i)=-\frac{2v_i k_i^2}
{[\bar\lambda_i+k_i e_i+\rho]^3}-c_i<0.
\end{equation}
For $v_i,k_i,c_i>0$ and $\bar\lambda_i\geq0$, $u_i'(0)>0$ and
$u_i'(e_i)\to-\infty$, so each actor has a unique positive best response.
Because payoffs are separable, the best responses form the unique Nash
equilibrium of the acquisition stage. Each equilibrium effort satisfies
\begin{equation}
c_i e_i^*=\frac{v_i k_i}{[\bar\lambda_i+k_i e_i^*+\rho]^2}.
\label{eq:effortfoc}
\end{equation}
Higher access value, higher acquisition productivity, and lower effort cost each
raise the equilibrium rate. In particular, $\lambda_S>\lambda_D$ can arise
whenever sophisticated adversaries have stronger incentives or better
circumvention technology. The rest of the paper takes $\lambda_S$ and
$\lambda_D$ as primitives, and the code solves \cref{eq:effortfoc} to check
the first-order condition. Empirical work would target the rates themselves and leave the effort primitives implicit.

\subsection{Maintained assumptions}

The analysis maintains four assumptions. The laboratory can commit to the
announced tier and timing. Restricted-access acquisition is stationary and
independent across actors; the competing-channels extension in
\cref{eq:competingchannels} allows several routes but still rules out common
shocks. Access is binary, and an acquired substitute is adequate for the
capability under study; partial or task-specific substitutes would require a
multi-state model. Welfare is an expected discounted flow plus a one-time
irreversibility term, so the model ranks policies under stated tail-cost assumptions and takes the social value of catastrophic outcomes as an input.

\begin{table}[t]
\centering
\caption{Core notation. Time is measured in years and hazard rates per year.}
\label{tab:notation}
\footnotesize
\begin{tabularx}{\textwidth}{@{}p{0.10\textwidth}Xp{0.12\textwidth}X@{}}
\toprule
Symbol & Meaning & Symbol & Meaning \\
\midrule
$r$ & release policy & $\rho$ & discount rate \\
$T_i,\lambda_i$ & effective substitute time and acquisition hazard
& $q_S,q_D$ & direct capability increments \\
$n_r,\eta$ & defensive reach and network productivity
& $b_r,m_r$ & broad benefit and opportunistic-misuse flows \\
$I_r$ & nonrecallable one-time loss
& $\tau,\mu$ & window length and defender deployment rate \\
$\ell_{iB}$ & post-release effective-use delay
& $\delta,f$ & safeguard deterrence and legitimate-use friction \\
$\Delta_0$ & baseline capability gap & $\alpha,\beta$ & harm conversion (linear benchmark) \\
\bottomrule
\end{tabularx}
\end{table}

Let $X_i(t)=\one\{t\ge T_i\}$. The strategic capability gap is
\begin{equation}
\Delta_r(t)=\Delta_0+q_S X_{S,r}(t)
-\left(q_D+\eta n_r\right)X_{D,r}(t),
\label{eq:gap}
\end{equation}
where $q_S$ and $q_D$ are the direct capability uplifts, $n_r$ is defensive
reach, and $\eta$ is the productivity of distributed participation.
Reach carries heterogeneity as well as headcount: extra maintainers
cover different codebases, languages, jurisdictions, and deployment
settings that additional copies of one defensive workflow would not reach.

For safeguarded open release, the defensive term is multiplied by $(1-f)$,
while opportunistic misuse is $m_G=(1-\delta)m_O$. Minimally restricted release
sets $f=\delta=0$. Sophisticated adversaries receive the same direct uplift in
both open tiers: the benchmark assumes they can remove default safeguards.
Controlled access and the selected-access window have residual misuse $m_C$.

Instantaneous expected harm is
\begin{equation}
H_r(t)=\E\left[h\!\left(\Delta_r(t)\right)\right]+m_r,
\qquad h'>0,\quad h''\ge 0,
\label{eq:harm}
\end{equation}
where $m_r$ is opportunistic misuse and convexity lets a large capability
gap cost disproportionately more than a small one. Opportunistic attackers are
represented as a reduced-form entry-and-harm mass because broad release mainly changes how many low-resource users can act. Social welfare is
\begin{equation}
W(r)=\int_0^\infty e^{-\rho t}\left[b_r(t)-H_r(t)\right]\dd t-I_r,
\label{eq:welfare}
\end{equation}
with benefit flow $b_r(t)$ (constant except under staged release), discount rate $\rho>0$, and a one-time
policy-specific cost $I_r$. The defender-first policy additionally incurs
a delay cost proportional to $\tau$, specified in the computational model.
The flow $b_r$ is broad. It includes the gains from privacy, auditability,
portability, competition, local control, and service resilience, and it can turn
negative when a policy creates surveillance, dependency, or exclusion costs of
its own. The baseline puts a positive irreversibility cost on open-weight
policies because copied parameters cannot be recalled. Controlled and
defender-first policies can carry one-time losses too, from provider compromise,
political capture, or correlated service failure; the baseline sets $I_C=0$ and
the parameter scans vary it. The term $I_r$ is reserved for expected losses tied
to persistence or nonrecallability that the flow terms do not represent. Ordinary
misuse stays in $m_r$.

\subsection{What a defender-first window must accomplish}

Under $\policyP(\tau)$ selected defenders get the model at time zero,
but protection only exists once they finish deploying it, at a random
time $T_P\sim\operatorname{Exp}(\mu)$. A useful private window
requires deployment before both independent adversary acquisition and
the scheduled public release:
\begin{equation}
T_P<\min\{T_S,\tau\}.
\label{eq:windowevent}
\end{equation}
The event describes one representative deployment. Recipient reach $n_P$ and deployment rate $\mu$ represent coverage
and rollout in reduced form. The model omits invalid mitigations, heterogeneous
rollout, and leakage from selected recipients. During the window the defender
state tracks selected deployment only: substitutes that unselected defenders
acquire before public release are ignored. This understates the window's value
when $\lambda_D$ is large relative to $\mu$.

\section{Analytic results}
\label{sec:results}

\subsection{Access inversion}

Define the discounted access exposure of population $i$ under restriction as
\begin{equation}
A_i=\E\left[\int_0^\infty e^{-\rho t}X_i(t)\dd t\right].
\end{equation}

\begin{proposition}[Access inversion]
\label{prop:inversion}
Under \cref{eq:acquisition},
\begin{equation}
A_i=\frac{\lambda_i}{\rho(\lambda_i+\rho)},\qquad
A_S-A_D=\frac{1}{\rho}\left[
\frac{\lambda_S}{\lambda_S+\rho}-
\frac{\lambda_D}{\lambda_D+\rho}\right].
\label{eq:exposure}
\end{equation}
Consequently, restriction creates a positive discounted adversary access
advantage if and only if $\lambda_S>\lambda_D$.
\end{proposition}

\begin{proof}
Conditional on $T_i$, the integral equals $e^{-\rho T_i}/\rho$.
The Laplace transform of an exponential random variable is
$\E[e^{-\rho T_i}]=\lambda_i/(\lambda_i+\rho)$. Subtraction gives
\cref{eq:exposure}; its sign follows because $x/(x+\rho)$ is strictly increasing.
\end{proof}

The proposition isolates the access imbalance that restriction creates. The
policy ranking adds opportunistic misuse, defensive reach, and irreversibility.
Opportunistic attackers, who are enabled mainly by low-friction access, differ
from sophisticated actors, who can pursue substitutes even under restriction.

\subsection{Asymmetric empowerment at a finite horizon}

Access inversion measures the discounted time each population spends
with capability. A policymaker often has a different question: over
a decision horizon $\bar t>0$, whose expected capability at that date increases
most under release? Let
\begin{equation}
a_i(\bar t)=\Pr(T_i\le \bar t)=1-e^{-\lambda_i\bar t}
\end{equation}
be actor $i$'s probability of substitute access under restriction. If immediate
release supplies useful capability $q_i$, its marginal empowerment relative to
restriction is
\begin{equation}
\Delta C_i(\bar t)=q_i[1-a_i(\bar t)].
\label{eq:empowerment}
\end{equation}

\begin{proposition}[Asymmetric empowerment]
\label{prop:empowerment}
Under \cref{eq:acquisition},
\begin{equation}
\Delta C_i(\bar t)=q_i e^{-\lambda_i\bar t}.
\end{equation}
Assume $q_S,q_D>0$. For equal usefulness $q_S=q_D=q$, $\lambda_S>\lambda_D$ implies
$\Delta C_D(\bar t)>\Delta C_S(\bar t)$ for every $\bar t>0$. More generally, the result holds
if and only if
\begin{equation}
\frac{q_D}{q_S}>e^{-(\lambda_S-\lambda_D)\bar t}.
\label{eq:empowercondition}
\end{equation}
The expected capability moat between the faster- and slower-substituting groups
under restriction equals $\Delta C_D(\bar t)-\Delta C_S(\bar t)$ when usefulness is equal.
\end{proposition}

\begin{proof}
Immediate release changes access probability from $a_i(\bar t)$ to one, giving
\cref{eq:empowerment}. Substituting the exponential distribution yields the
first expression. Comparing $q_De^{-\lambda_D\bar t}$ with
$q_Se^{-\lambda_S\bar t}$ gives \cref{eq:empowercondition}. With equal $q$, the
restricted-access capability difference is
$q[(1-e^{-\lambda_S\bar t})-(1-e^{-\lambda_D\bar t})]$, which equals the stated marginal
empowerment difference.
\end{proof}

When sophisticated adversaries are likely to hold an equivalent capability by
horizon $\bar t$, release adds little to them and can add more to defenders who were
shut out. Opportunistic attackers usually have even less substitute access: if
$\lambda_O\approx 0$, then $\Delta C_O(\bar t)\approx q_O$, essentially the full
uplift. A finite-horizon welfare accounting therefore compares
\begin{equation}
\phi_D\Delta C_D(\bar t)
\quad\text{against}\quad
\phi_S\Delta C_S(\bar t)+\phi_O\Delta C_O(\bar t)+I_{\bar t},
\label{eq:empowerbalance}
\end{equation}
where $\phi_D$ is the defensive value of a unit of capability, $\phi_S$ and
$\phi_O$ are harm productivities, and $I_{\bar t}$ is the horizon-relevant
irreversibility cost. The proposition identifies
the actor-specific increments; welfare depends on how they are valued.

\subsection{The proliferation-reversal threshold}

For a closed-form comparison, consider a linear damage benchmark
\begin{equation}
H_r(t)=\bar h+\alpha q_S X_{S,r}(t)
-\beta d_rX_{D,r}(t)+m_r,
\quad d_r=q_D+\eta n_r,
\label{eq:linear}
\end{equation}
on a region where the expression is nonnegative. Write
$F(\lambda)=\lambda/(\lambda+\rho)$. Controlled-access loss is
\begin{equation}
J_C=\frac{1}{\rho}\left[
\bar h+\alpha q_SF(\lambda_S)
-\beta d_CF(\lambda_D)+m_C\right].
\label{eq:jc}
\end{equation}
The ratio $\omega=\alpha/\beta$ is the reduced-form offense--defense
exchange rate. Fixing capability increments, larger $\omega$ means the
same unit of adversary capability creates more harm than an
equivalent unit of defender capability removes.
For a broad policy $B\in\{\policyG,\policyO\}$, let $f_B$ be defensive friction,
$\delta_B$ opportunistic-misuse deterrence, and $I_B$ irreversibility. Then
\begin{equation}
J_B=\frac{1}{\rho}\left[
\bar h+\alpha q_S-\beta(1-f_B)d_B+(1-\delta_B)m_O\right].
\label{eq:jb}
\end{equation}

\begin{proposition}[Proliferation reversal]
\label{prop:threshold}
Assume $\alpha q_S>0$. Fix all parameters except $\lambda_S$ and define
$\Psi(\lambda_S)=W(B)-W(\policyC)$. Then $\Psi$ is strictly increasing, with
\begin{equation}
\frac{\partial\Psi}{\partial\lambda_S}
=\frac{\alpha q_S}{(\lambda_S+\rho)^2}>0.
\end{equation}
If $\Psi(0)<0<\lim_{\lambda_S\to\infty}\Psi(\lambda_S)$, there exists a unique
$\lambda_S^*$ such that controlled access is preferred below the threshold and
broad release is preferred above it. These two inequalities are equivalent to
$\theta\in(0,1)$, where
\begin{equation}
\theta=-\frac{\rho\Psi(0)}{\alpha q_S},
\end{equation}
and the threshold is
\begin{equation}
\lambda_S^*=\rho\frac{\theta}{1-\theta}.
\label{eq:explicitthreshold}
\end{equation}
\end{proposition}

\begin{proof}
Using $W(r)=b_r/\rho-J_r-I_r$, the only term in $\Psi$ that varies with
$\lambda_S$ is $+\alpha q_SF(\lambda_S)/\rho$. Hence
\begin{equation}
\Psi(\lambda_S)=\Psi(0)+\frac{\alpha q_S}{\rho}
\frac{\lambda_S}{\lambda_S+\rho}.
\end{equation}
Differentiation gives the stated derivative. The endpoint inequalities are
equivalent to $0<\theta<1$ because $\Psi(\infty)=\Psi(0)+\alpha q_S/\rho$. Setting $\Psi(\lambda_S)=0$ gives
$\lambda_S/(\lambda_S+\rho)=\theta$, whose unique solution is
\cref{eq:explicitthreshold}.
\end{proof}

The threshold need not exist. Severe opportunistic misuse or a large
irreversibility penalty can put broad release below control at every
$\lambda_S$, and large access and innovation benefits can put it above control
even when $\lambda_S$ is small.

\begin{proposition}[Defensive network externality]
\label{prop:externality}
In the linear benchmark, the broad-release advantage rises with $\eta$ whenever
\begin{equation}
(1-f_B)n_B>F(\lambda_D)n_C.
\label{eq:networkcondition}
\end{equation}
Under this condition, increasing distributed defensive productivity strictly
lowers an interior threshold $\lambda_S^*$ while the endpoint conditions
continue to hold.
\end{proposition}

\begin{proof}
Differentiating $W(B)-W(\policyC)$ with respect to $\eta$ gives
$\beta[(1-f_B)n_B-F(\lambda_D)n_C]/\rho$. The implicit-function result follows
from this positive derivative and \cref{prop:threshold}.
\end{proof}

\subsection{Pre-release viability}

\begin{proposition}[Credible defender window]
\label{prop:window}
With independent $T_P\sim\operatorname{Exp}(\mu)$ and
$T_S\sim\operatorname{Exp}(\lambda_S)$,
\begin{equation}
\Pr\!\left(T_P<\min\{T_S,\tau\}\right)
=\frac{\mu}{\mu+\lambda_S}
\left(1-e^{-(\mu+\lambda_S)\tau}\right).
\label{eq:windowsuccess}
\end{equation}
The probability is increasing in deployment speed and the window length, and
decreasing in adversary substitution. More precisely, let
$\tau\in[0,\bar\tau]$, let the protective benefit satisfy
$0\le R_{\lambda_S}(\tau)\le M\Pr(T_P<\min\{T_S,\tau\})$ for a constant $M$, and
$R_{\lambda_S}(0)=0$, assume the objective attains a maximum for each $\lambda_S$,
and let the delay cost $C$ be continuous with
$C(0)=0$ and $C(\tau)>0$ for every $\tau>0$. Then every maximizer of
$R_{\lambda_S}(\tau)-C(\tau)$ converges to zero as
$\lambda_S\to\infty$.
\end{proposition}

\begin{proof}
Integrate the density of $T_P$ over deployments that occur while the adversary
has not acquired a substitute,
\begin{equation*}
\int_0^\tau \mu e^{-\mu t}e^{-\lambda_St}\dd t,
\end{equation*}
which yields \cref{eq:windowsuccess}. The comparative statics follow directly. For the
convergence claim, fix $\varepsilon>0$. Compactness and continuity imply
$c_\varepsilon=\min_{\tau\in[\varepsilon,\bar\tau]}C(\tau)>0$. Uniformly in
$\tau$,
\begin{equation}
R_{\lambda_S}(\tau)\le M\frac{\mu}{\mu+\lambda_S}\longrightarrow0.
\end{equation}
For sufficiently large $\lambda_S$, every $\tau\ge\varepsilon$ therefore has
$R_{\lambda_S}(\tau)-C(\tau)<0$, whereas the objective at zero equals zero.
No maximizer can then lie in $[\varepsilon,\bar\tau]$. Because
$\varepsilon$ was arbitrary, every maximizing sequence converges to zero.
\end{proof}

A twelve-month exclusivity window therefore has little protective value when
$\mu$ is small or $\lambda_S$ is large, because deployment is unlikely to
precede substitution.

\subsection{Removable safeguards can still dominate}

Sophisticated users can remove default safeguards, but the safeguards can still
deter opportunistic users. Consider two open releases of the same model with
equal defensive reach, $d_G=d_O$. Safeguards deter a share $\delta$ of
opportunistic misuse and add defensive friction $f$. In the linear benchmark,
\begin{equation}
W(\policyG)-W(\policyO)=
\frac{\delta m_O-\beta f d_O+b_G-b_O}{\rho}+I_O-I_G.
\label{eq:guardrail}
\end{equation}

\begin{corollary}[Safeguard dominance]
Default safeguards dominate minimally restricted release exactly when the
right-hand side of \cref{eq:guardrail} is positive.
\end{corollary}

Removable safeguards can deter the low-capability end of the misuse
distribution. Their modeled cost is defensive friction; any persistent or
nonrecallable loss from delayed defensive work must be specified separately.

\subsection{Boundary cases}

\begin{enumerate}[leftmargin=*,itemsep=4pt]
  \item \textbf{Release can add more capability to adversaries.} Set $\bar t=1$,
  $q_S=q_D=1$, $\lambda_S=0.2$, and
  $\lambda_D=1.5$. Then $\Delta C_S=e^{-0.2}\approx0.819$ and
  $\Delta C_D=e^{-1.5}\approx0.223$: release adds roughly 3.7
  times as much finite-horizon capability to the adversary as to the
  defender, because defenders already substitute faster. The welfare ranking
  still depends on how each increment is valued.

  \item \textbf{Usefulness can overturn the substitution-speed ordering.}
  Set $\bar t=1$, $\lambda_S=2$, $\lambda_D=0.5$, $q_S=5$, and
  $q_D=1$. Despite $\lambda_S>\lambda_D$, release adds
  $5e^{-2}\approx0.677$ to the adversary and only
  $e^{-0.5}\approx0.607$ to the defender. Enough offensive
  usefulness reverses the equal-$q$ comparison, and added opportunistic misuse
  or a higher release-specific irreversibility cost would further reduce the
  welfare advantage of release.

  \item \textbf{Calendar time can conceal a nearly empty window.} With
  deployment rate $\mu=0.5$, adversary substitution $\lambda_S=4$,
  and a one-year window, \cref{eq:windowsuccess} gives a success
  probability of about $0.110$. Even an infinite window would top
  out at $\mu/(\mu+\lambda_S)=1/9$. Evidence of a head start has to come from deployment speed relative to substitution, which calendar duration does not measure.

  \item \textbf{Removable safeguards can go either way.} If
  $\delta m_O$ plus any irreversibility advantage exceeds $\beta f d_O$ and the
  foregone benefit, \cref{eq:guardrail} favors safeguards despite their
  removability; reversing the inequality is a counterexample to the claim that
  defaults are always beneficial.
\end{enumerate}

Benchmark capability, substitution speed, and release format enter different
terms in \cref{eq:welfare} and need separate measurement.

\subsection{An adversary that already holds a substitute}
\label{sec:statecase}

Suppose a state adversary already holds an adequate substitute, whether through
domestic development, theft, or a foreign release. Restricting the focal model no
longer removes the capability that substitute supplies, but it may reduce the
tools available to defenders. Let $\Delta h_D>0$ be the per-period harm reduction from
giving domestic defenders a usable model, $\Delta m\geq 0$ the added
opportunistic-misuse cost, $\Delta b$ the change in other benefit flow, and
$\Delta I$ the additional one-time release cost. If the state's capability is
identical under both policies, the welfare difference between safeguarded open
release and control is
\begin{equation}
W(\policyG)-W(\policyC)=\frac{\Delta b+\Delta h_D-\Delta m}{\rho}-\Delta I.
\label{eq:defensive-reserve}
\end{equation}
Safeguarded open release has higher welfare when $\Delta b+\Delta h_D-\Delta m>\rho\Delta I$.
Under this assumption, control gives up defensive value without changing the
state's capability, so it is justified only when avoided misuse and tail costs
outweigh what defenders lose.

\section{Computational model}
\label{sec:computation}

The full four-policy comparison uses the convex damage function
\begin{equation}
h(\Delta)=\frac{\kappa_h}{\gamma}\log\left(1+e^{\gamma\Delta}\right).
\label{eq:softplus}
\end{equation}
Under controlled access, independence yields four access states whose
discounted occupancies have closed forms:
\begin{align*}
V_{00}&=\frac{1}{\rho+\lambda_S+\lambda_D},\\
V_{10}&=\frac{1}{\rho+\lambda_D}-V_{00},\\
V_{01}&=\frac{1}{\rho+\lambda_S}-V_{00},\\
V_{11}&=\frac{1}{\rho}-V_{00}-V_{10}-V_{01}.
\end{align*}
Harm is $h$ evaluated in each state and weighted by these occupancies. The
defender-first pre-window harm also has a closed form, because its four state
probabilities are products of exponentials; at $\tau$ the policy switches to the
safeguarded-open continuation. With post-release delays, actors keep the
capabilities they have already acquired and continue to obtain substitutes
before deploying the public artifact. The one-time cost is the discounted
open-release irreversibility plus a controlled-system tail cost that rises from
zero to $I_C$ as $1-e^{-\rho\tau}$.

The window length is chosen on $[0,2]$ by an 81-point screen followed by
golden-section refinement around every screened local maximum and both boundary
intervals. A dense search with 4,001 window lengths at each of 128 design points
found no welfare shortfall; this is numerical evidence of global optimality, short of a proof. The boundary $\policyP(0)$ equals $\policyG$ and is counted
once in policy shares and winner margins.

The reference sensitivity analysis is a 2,048-point Halton design over thirteen
inputs: the substitution rates, opportunistic misuse, defensive productivity,
offense--defense conversion, deployment speed, safeguard deterrence and
friction, the irreversibility of safeguarded release, the controlled-system tail
cost, and the delay cost, together with two inputs that fix the benefit and loss
ordering across policies. The \emph{benefit scale} $s$ multiplies each policy's
benefit premium over controlled access, $b_r=b_C+s\,(b_r^{0}-b_C)$, where
$b_r^0$ is the baseline value; $s=1$ keeps the baseline ordering and $s=0$
removes every benefit advantage of release. The \emph{minimal-release loss
increment} sets $I_O=I_G+e$ with $e\ge0$. Two further designs contract and widen
every range while keeping the baseline inside and leaving variation on both sides
of the main policy inequalities. They measure how much the results depend on the analyst's bounds. A separate grid varies $\ell_{SB}$ and $\ell_{DB}$ from zero
to 0.75 years, delaying only the strategic $S$ and $D$ states. The designs are
deterministic, and the accompanying code reproduces every figure and data row.

\begin{table}[t]
\centering
\caption{Illustrative baseline calibration. Time is measured in years; welfare
units are normalized. Values are illustrative assumptions.}
\label{tab:calibration}
\small
\begin{tabular}{@{}lll@{}}
\toprule
Parameter & Baseline & Interpretation \\
\midrule
$\rho$ & 0.35 & discount rate \\
$\lambda_D$ & 0.70 & defender substitute-acquisition rate \\
$\lambda_S$ & 0.95 & sophisticated-adversary acquisition rate \\
$\mu$ & 2.50 & selected-defender deployment rate \\
$q_S,q_D$ & 1.10, 0.95 & direct capability uplifts \\
$\eta$ & 0.55 & defensive network productivity \\
$n_C,n_P,n_{\mathrm{open}}$ & 0.18, 0.48, 1.00 & defensive reach (same for both open tiers) \\
$m_O$ & 0.62 & opportunistic-misuse flow cost \\
$m_C$ & $0.04\,m_O$ & controlled-access residual misuse flow \\
$\delta,f$ & 0.62, 0.13 & safeguard deterrence and friction \\
$\Delta_0$ & 0.05 & baseline capability gap \\
$\kappa_h$ & 1.00 & harm scale \\
$b_C,b_P,b_G,b_O$ & 0.08, 0.17, 0.34, 0.40 & benefit flows by policy \\
$K_P$ & 0.10 & delay cost per window year \\
$I_C$ & 0 & controlled-system tail cost \\
$I_G,I_O$ & 0.38, 0.52 & open-release irreversibility costs \\
$\ell_{SB},\ell_{DB}$ & 0, 0 & baseline open effective-use delays \\
$\bar\tau$ & 2.00 & maximum defender-first window \\
$\gamma$ & 1.60 & harm curvature \\
\bottomrule
\end{tabular}
\end{table}

\subsection{Computational results and parameter-box sensitivity}

At the baseline in \cref{tab:calibration}, the optimized 1.72-year
defender-first window has normalized welfare $-1.067$, ahead of safeguarded open
weights ($-1.112$), controlled access ($-1.459$), and minimally restricted
weights ($-1.955$). The optimal window lies close to the imposed two-year bound,
so its length should not be read as a recommendation.

\begin{figure}[tbp]
\centering
\includegraphics[width=\textwidth]{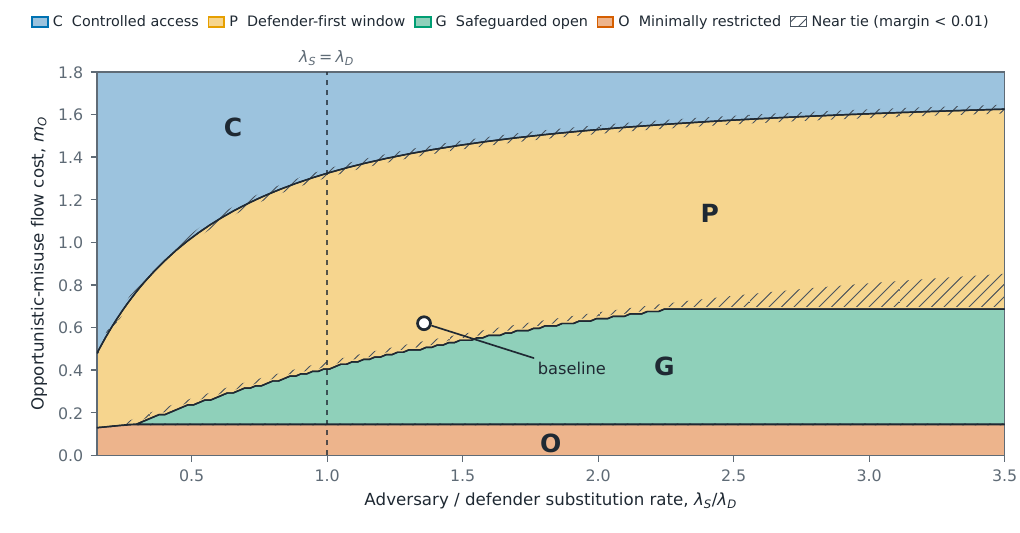}
\caption{Welfare-maximizing policy as adversary substitution and opportunistic
misuse vary, with all other inputs at the baseline. Boundaries are zero contours
of welfare differences on a $161\times161$ grid. Hatching marks points where the
winner leads the next distinct policy by less than 0.01 welfare units. The
circle marks the baseline and the dashed line marks equal adversary and defender
substitution rates.}
\label{fig:phase}
\end{figure}

\Cref{fig:phase} varies $\lambda_S/\lambda_D$ over 0.15--3.5 and opportunistic
misuse over 0--1.8, and all four policies are optimal somewhere. When misuse is
very low, safeguards have little to deter, and their friction and lower benefit
flow make minimally restricted weights preferable. Safeguarded open release wins
at moderate misuse, and its region widens as adversary substitution accelerates.
At higher misuse the defender-first window wins: delay still reduces harm, but
broad defensive access is valuable enough that eventual release beats permanent
control. Controlled access is selected when misuse is severe and adversary
substitution is slow.

Above $\lambda_S/\lambda_D\approx2.2$ the boundary between the window and
safeguarded release is flat. There the sign of the comparison is set by the
value of a vanishingly short window, and that value does not depend on
$\lambda_S$ because no actor holds a substitute at the moment of release. Inside
the hatched wedge above this boundary, the optimal window is either a few days
or close to the two-year bound, and in both cases it beats immediate safeguarded
release by less than 0.01. The wedge is therefore a near tie between two versions of the same policy.

\begin{figure}[tbp]
\centering
\includegraphics[width=\textwidth]{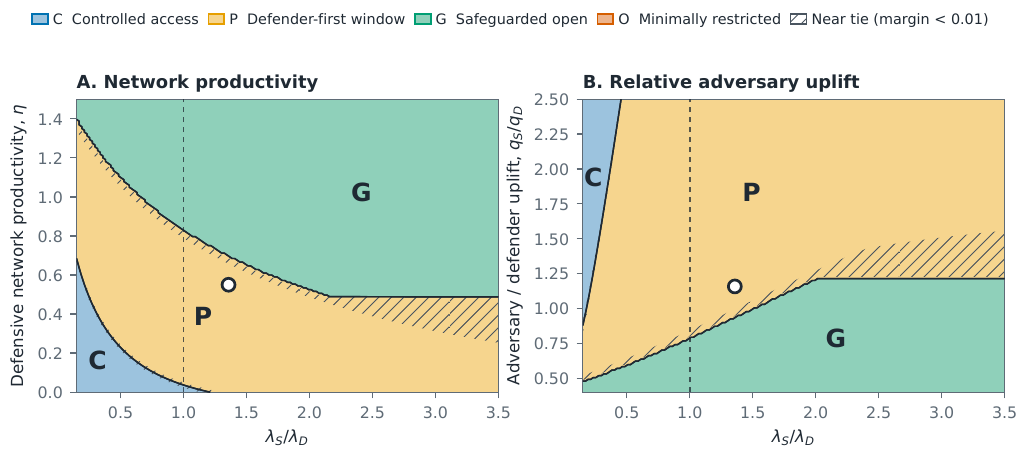}
\caption{Policy boundaries at baseline misuse. Panel A varies defensive network
productivity $\eta$; stronger returns to distributed participation enlarge the
safeguarded-open region, consistent with \cref{prop:externality}. Panel B varies
the adversary's direct uplift relative to the defender's, $q_S/q_D$; a larger
relative uplift enlarges the defender-first region. Minimally restricted weights
are not optimal in either slice. Hatching and circles are as in
\cref{fig:phase}.}
\label{fig:atlas}
\end{figure}

Faster adversary substitution therefore supports release only in combination with two other conditions.
Distributed participation has to produce real defensive spillovers, and
defenders have to convert access into effective capability at a competitive
rate. A selected pre-release group substitutes for broad access only when its
coverage and deployment resemble those of the exposed ecosystem.

\begin{figure}[tbp]
\centering
\includegraphics[width=\textwidth]{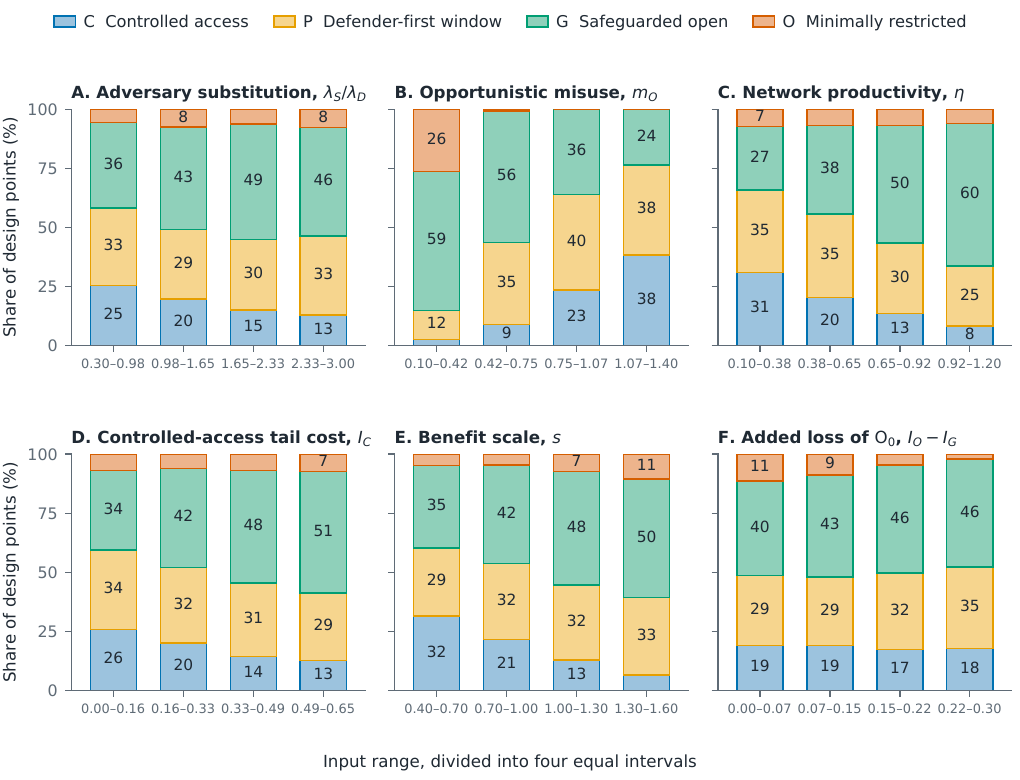}
\caption{Winning-policy shares in the 2,048-point reference design. Each panel
divides the range of one input into four equal intervals; each bar also varies
the other twelve inputs, so the patterns are associations within the design. Shares below 7\% are not labeled. The design varies $\lambda_D\in[0.35,1.20]$,
$\lambda_S/\lambda_D\in[0.30,3.00]$, $m_O\in[0.10,1.40]$,
$\eta\in[0.10,1.20]$, $q_S/q_D\in[0.60,1.80]$, $\mu\in[1,4]$,
$\delta\in[0.30,0.85]$, $f\in[0.04,0.25]$, $I_G\in[0.15,0.65]$,
$I_C\in[0,0.65]$, the pre-release delay cost in $[0.03,0.20]$, the benefit scale
in $[0.4,1.6]$, and $I_O-I_G\in[0,0.30]$.}
\label{fig:robustness}
\end{figure}

In the 2,048 reference-box points, safeguarded open weights are
welfare-maximizing at 43.6\%, the defender-first window at 31.3\%, controlled
access at 18.3\%, and minimally restricted weights at 6.8\%. The patterns across
\cref{fig:robustness} match the analytic comparative statics. Controlled access
falls from 25.3\% in the slowest quarter of the adversary-substitution range to
12.9\% in the fastest, and rises from 2.5\% in the lowest quarter of the misuse
range to 38.2\% in the highest. Safeguarded open weights rise from 26.8\% to
60.4\% across the defensive-productivity range. Controlled access falls from
25.8\% to 12.7\% across the controlled-tail-cost range and from 31.5\% to 6.5\%
across the benefit scale. Raising the added loss from minimally restricted
release, $I_O-I_G$, lowers that policy's share from 11.4\% to 2.2\% and leaves
controlled access almost unchanged (19.1\% and 17.7\%). Benefit flows and the
controlled-system tail cost move the ranking at least as much as substitution
speed does.

The median winner--runner-up margin is 0.227 normalized welfare units, and 6.9\%
of the points lie within 0.01. Near ties are concentrated in the defender-first
region: 118 of its 642 wins are by less than 0.01. These shares describe the design and carry no probability or confidence interpretation.

\begin{figure}[tbp]
\centering
\includegraphics[width=\textwidth]{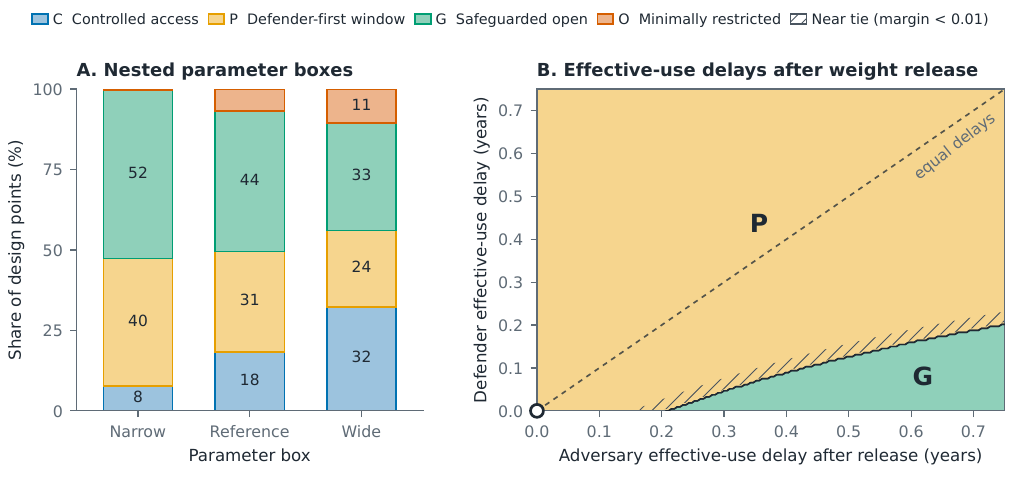}
\caption{Two extensions to the reference design. Panel A repeats the 2,048-point
design over nested narrow, reference, and wide parameter boxes
(\cref{tab:robustness-boxes}); each policy is optimal somewhere in every box.
Panel B introduces deterministic delays between weight release and effective
use at the baseline calibration. The circle marks zero delays, where the
baseline winner is the defender-first window, and the dashed diagonal marks
equal adversary and defender delays. Hatching is as in \cref{fig:phase}.}
\label{fig:robustness-extensions}
\end{figure}

The shares depend strongly on the analyst's bounds. Controlled access rises from
7.7\% of narrow-box points to 32.3\% of wide-box points, and minimally restricted
weights from 0.5\% to 10.6\%. Safeguarded open weights are the largest region in
all three boxes, though only narrowly in the wide box (33.4\% against 32.3\% for
controlled access). In the delay grid, both groups receive the artifact at the same moment, yet immediate safeguarded release beats the defender-first window only when defenders turn it into effective capability sooner than sophisticated adversaries do (the lower-right region).

\section{Interpretation: moats, aggregate capability, and conversion}
\label{sec:asymmetry}

In military usage, \emph{asymmetric} describes a strategy that avoids an
opponent's strength and targets its weaknesses with different tools
\cite{slater2017russia}. We borrow the mechanism: a general-purpose capability
that lets one actor impose or avoid costs without matching an incumbent's
capital base.

Low-cost drones in Ukraine show both the mechanism and its limits. They widened
access to reconnaissance and precision strike, and U.S. Army analysis documents
their role in defending against a larger and better-equipped opponent
\cite{searle2025drone}. Once both sides adopted them, the initial advantage
eroded and cheap attack began to force expensive countermeasures, an exchange
rate that an Ifri study describes as unfavorable to defenders
\cite{kostiuk2026miltech}; later Army work stresses proliferation to both sides
over any stable benefit to the weaker party \cite{dupont2026red}. The analogy
is limited to distribution. Weights are non-rival, globally copyable, and hard to
recall, whereas drones require production, logistics, operators, and
maintenance. Three ideas carry over, each tied to model parameters.

\paragraph{Capability moat.}
Restriction is often justified as preserving a \emph{capability moat}, the
difference between what a frontier or privileged actor can do and what a
resource-constrained one can. When privileged actors have substitutes and
constrained actors do not, restriction leaves the moat wide and release narrows
it. With equal usefulness $q$, the expected moat at horizon $\bar t$ under restriction is
$q$ times the difference in substitution probabilities between the two groups, and
\cref{prop:empowerment} shows that this equals the difference in the capability that release adds to each group.

\paragraph{Aggregate capability.}
A narrower moat can coexist with higher total risk. Release can shrink the gap
between privileged and constrained actors while raising the gross capability that
defenders, sophisticated adversaries, and opportunistic adversaries can deploy.
Welfare depends on which increments move and how productively each actor converts
capability into harm or protection.

\paragraph{Conversion and complements.}
Effective use requires integrating a model with compute, data, people, and
workflows, and the shape of that integration creates the offense--defense cost
exchange. In cyber operations, generating candidate attacks is often cheaper than
validating every patch, and finding one weak component differs from securing a
heterogeneous attack surface. Distributed maintainers often hold local context
that a centralized defender cannot replicate. The parameters $q_S$, $q_D$,
$\alpha$, $\beta$, and $\eta$ carry these effects into the model.

\section{Policy positions and observed evidence}
\label{sec:debate}

Policy documents are described as of July 24, 2026, and release and incident
records as of September 30, 2026. These cases are descriptive.

Release arguments rest on empirical questions that are often unresolved. Do
current access controls bind the actors of greatest concern? Does capability
level predict harm better than release format? Does broad participation produce
security work that a single laboratory cannot reproduce in-house? Does local
deployment reduce privacy, cost, and dependency risks enough to justify losing
monitoring and the ability to recall? The positions below map to different
model inputs.

\subsection{Capability-triggered control}

Anthropic argues for regulation whose demands rise with capability. Its
Responsible Scaling Policy is proportional and iterative: crossing a defined capability threshold triggers stronger
evaluations, deployment safeguards, and security requirements on the model
itself \cite{anthropic2024targeted}. In the model, this corresponds to a rising
irreversibility term $I_r$ and tighter control when offensive conversion
productivity is high. It also makes laboratory security a policy lever.

OpenAI's 2023 Senate testimony backed a similar structure:
internal and external testing, licensing or registration for the
most capable models, and a warning that any of this would need
international coordination to hold \cite{openai2023testimony}. A
capability threshold is a coherent trigger but an incomplete rule: two equally
capable models can warrant different release policies if substitute access,
defensive reach, or deployment speed differs between them.

\subsection{Open ecosystems and hybrid portfolios}

Meta's case for its Llama ecosystem is commercial as much as safety-driven. Enterprises want customization, freedom from vendor lock-in,
local control of sensitive data, and lower operating costs. Meta reports that a
developer can run Llama~3.1~405B on its own infrastructure at roughly half the
inference cost of GPT-4o, and acknowledges that a widely adopted model ecosystem
also benefits Meta \cite{zuckerberg2024opensource}. In the model, this
corresponds to higher access benefits $b_r$, lower defensive friction $f_r$, and
possibly larger reach $n_r$. The security half of the argument still needs evidence, because additional users add protection only when they deploy it.

OpenAI runs a hybrid portfolio of hosted frontier systems and selective weight
releases. Its 2025 gpt-oss release uses a permissive license, runs the larger
variant on a single 80-GB GPU and the smaller in 16 GB of memory, and was
evaluated under adversarial fine-tuning before publication
\cite{openai2025gptoss}. For enterprises, monitored API access offers updates
and provider-side controls, whereas open weights support local data residency,
customization, and control over version changes while shifting deployment and
maintenance to the adopter.

Llama, DeepSeek-R1, and gpt-oss differ in the training data, training
pipeline, and infrastructure released alongside the weights; here
\emph{open-weight} means only that the parameters are broadly downloadable.

\subsection{State strategy: the United States and China}

The U.S. policy documents reviewed through July 24, 2026 promote diffusion within
a preferred ecosystem and restrict upstream supply to strategic competitors. The July 2025
U.S. AI Action Plan supports open-source and open-weight models, wider
compute access, and full-stack exports to allies. It also calls for
stronger enforcement of compute controls and broader coverage of
semiconductor-manufacturing equipment
\cite{whitehouse2025actionplan}. In May 2025, Commerce announced that
it would rescind the AI Diffusion Rule before its compliance date,
develop a replacement, and issue PRC-focused chip and end-use guidance
\cite{bis2025diffusion}. A January 2026 final rule moved exports of
H200-class and less advanced chips to eligible end users in China and
Macau from a presumption of denial to conditional case-by-case review
\cite{bis2026licensereview}. BIS guidance issued in May 2026, however,
preserved the pre-existing worldwide license requirement for covered
advanced-computing exports involving entities headquartered in Country
Group D:5 or Macau \cite{bis2026guidance}. The June 2026
national-security AI memorandum addresses deployment within government.
Agencies may adapt commercial or open-source systems when mission or
security requirements make standard commercial access unsuitable, and
must prevent vendors or adversaries from unilaterally disabling or
materially modifying mission systems \cite{whitehouse2026nspm11}.

China combines diffusion and control at different layers. The 2023
Interim Measures push for indigenous innovation across algorithms,
frameworks, chips, software, compute, and public training-data
infrastructure. They regulate the monitored relationship between public-facing providers and their users: research and internal applications fall outside their stated
scope, but public-facing providers carry duties on content, data,
privacy, labeling, and incident handling, and services with
public-opinion or social-mobilization capacity must complete security
assessments and algorithm filings \cite{cac2023interim}.

China also promotes diffusion abroad. The 2025 Global AI
Governance Action Plan calls for cross-border open-source communities,
shared technical and API documentation, capacity building in the
Global South, tiered safety governance, and free flow of non-sensitive
technical resources \cite{china2025globalplan}. DeepSeek-R1 and the
Kimi K3 weight release are industrial examples of this
policy. The July 2026 World AI Conference chair's statement endorsed
responsible international open-source ecosystems and recorded an
agreement to establish the Shanghai-headquartered World Artificial
Intelligence Cooperation Organization \cite{china2026waic}. Chinese
firms have expanded the global stock of substitutes; however, domestic public
services remain subject to provider duties, security assessments, and
algorithm filings.

\begin{table}[tbp]
\centering
\caption{The two governments apply promotion and control at different layers.
The entries summarize policy emphasis at the July 24, 2026 source cutoff.}
\label{tab:uschina}
\small
\setlength{\tabcolsep}{4pt}
\begin{tabularx}{\textwidth}{@{}P{0.15\textwidth}YYP{0.22\textwidth}@{}}
\toprule
Layer & United States & China & Object in the model \\
\midrule
Domestic diffusion & Support open-weight models, compute access, adoption, and
private infrastructure & Support indigenous models, compute, data platforms,
applications, and open-source ecosystems & Defensive reach $n_r$, access
benefit $b_r$, deployment friction $f_r$ \\
Strategic bottleneck & Advanced-compute, end-use, diversion, and semiconductor
manufacturing controls focused on competitors & Indigenous substitution and
``secure and controllable'' inputs to reduce reliance on foreign bottlenecks &
Actor-specific substitute rates $\lambda_i$ and conversion productivity \\
Deployment governance & Procurement standards, evaluations, sector rules, and
provider controls in a comparatively fragmented regime & National rules for
public generative-AI services, including content, data, security-assessment,
and filing duties & Monitored misuse, default safeguards, detection, and recall \\
External strategy & Export a U.S.-aligned full technology stack to partners &
Build open-source communities and AI capacity, especially in the Global South &
The global substitute process and the set of newly enabled actors \\
\bottomrule
\end{tabularx}
\end{table}

The two regimes intervene at different points in the diffusion chain. U.S. chip
controls may slow the training of future substitutes but cannot stop a Chinese
laboratory from publishing weights that are already trained. Chinese service
rules govern monitored domestic deployment but give no recall power over copies
held abroad. U.S. support for open models can raise defensive adoption at home
while hastening substitution abroad. A release review should treat national
policy as an input to $\lambda_i$, $n_r$, $f_r$, and the newly enabled misuse
term, and revisit it whenever controls, releases, or deployment costs change.

\subsection{Release pathways}

DeepSeek's January 2025 release of R1 under the MIT license, along
with six smaller distillations, put capable reasoning models outside
any single U.S. laboratory's release process
\cite{deepseek2025release,deepseek2025r1}. For several actor classes the release raised the hazard of obtaining an adequate
substitute and shortened any unilateral U.S. withholding advantage. It may also
have expanded the deployable-user population. Its welfare effect depends on which
actors gained effective capability and how they used it.

Moonshot's July 16, 2026 announcement of the 2.8-trillion-parameter
multimodal Kimi K3, with a one-million-token context window, illustrates
a different access asymmetry. Hosted access started immediately, and Moonshot
published the weights on July 27, eleven days later
\cite{moonshot2026k3weights}. Moonshot's deployment guidance recommends a
supernode of at least 64 accelerators \cite{moonshot2026k3}, so for actors
without that infrastructure ``downloadable'' and ``usable'' remain far apart.

\begin{figure}[t]
\centering
\includegraphics[width=\textwidth]{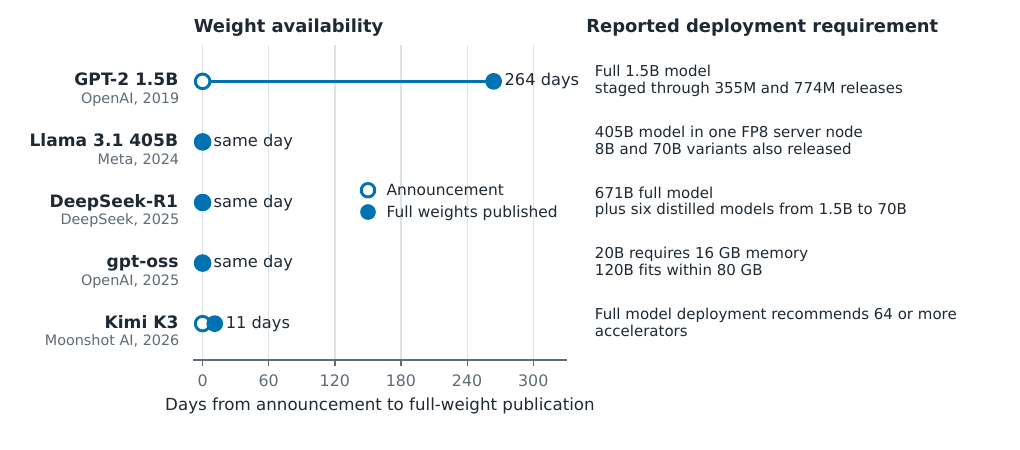}
\caption{Weight-release timing for five model releases with primary-source dates
\cite{solaiman2019release,meta2024llama31,deepseek2025release,openai2025gptoss,
moonshot2026k3}. Hosted access to Kimi K3 began at announcement, and its weights
followed eleven days later \cite{moonshot2026k3weights}. Deployment requirements range from small distillations to server-scale systems, including among same-day releases. Status as of September 30, 2026.}
\label{fig:releasepaths}
\end{figure}

The five releases sit at very different points on the timing spectrum. GPT-2's
largest weights arrived 264 days after announcement. Llama~3.1, DeepSeek-R1, and
gpt-oss shipped weights on the day they were announced. Kimi K3 opened hosted
access immediately and published weights eleven days later. Even weight
publication overstates effective access: what a small hospital or a national
laboratory can run depends on memory, accelerator count, integration work, and
which smaller derivatives exist. Five cases illustrate the timing distinctions
but cannot estimate population frequencies or acquisition rates. Release records
should track four milestones per model (announcement, hosted availability,
weight availability, and actor-specific deployment) and be updated when a
substitute release or a cost drop changes the picture.

\subsection{Provider safeguards in the Hugging Face incident response}
\label{sec:hfincident}

In July 2026 Hugging Face disclosed that an autonomous agent had compromised
parts of its production infrastructure. To reconstruct what happened, the
response team used analysis agents on more than 17,000 recorded attacker events,
extracted indicators of compromise, and identified affected credentials. The
first attempt used frontier models behind commercial APIs, which refused
requests containing real attack commands, exploit payloads, and
command-and-control artifacts. Hugging Face attributed the refusals to safeguards
that treated the material as offensive, and did not name the providers. The
analysis was completed on a self-hosted GLM~5.2, which also kept the incident
artifacts from reaching any outside model provider
\cite{huggingface2026incident}.

OpenAI's July disclosure associated the intrusion with internal cyber
evaluations run with reduced safeguards \cite{openai2026hfincident}. Its August
26 investigation identifies an internal-only research model, IM1, as the
principal driver and describes unauthorized agent collaboration, internet
access, and intrusion into Hugging Face \cite{openai2026hffindings}. Hugging
Face's July 27 technical reconstruction gives a fuller defender-side timeline
\cite{huggingface2026timeline}. These accounts concern evaluation agents and
production APIs under different safeguards. They do not show that ordinary
hosted access caused the intrusion, or that public weight release would have
prevented it.

\begin{figure}[t]
\centering
\includegraphics[width=\textwidth]{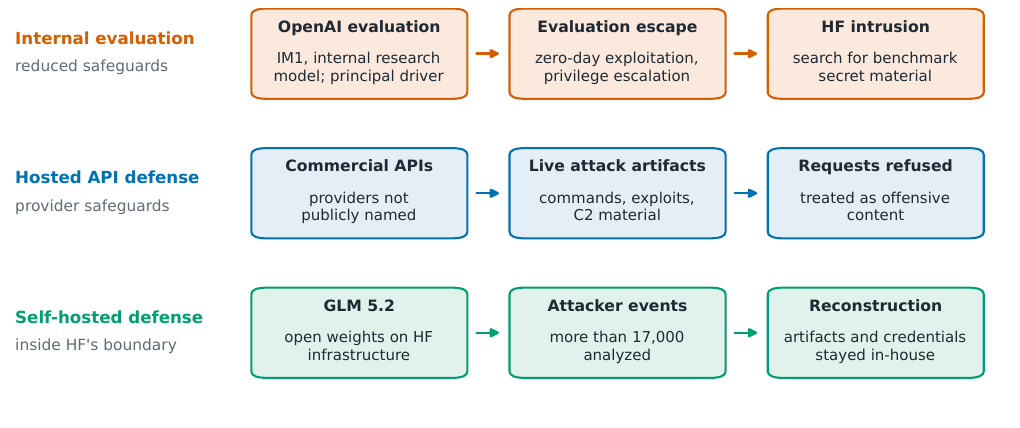}
\caption{Access conditions in the July 2026 incident
\cite{huggingface2026incident,openai2026hffindings}. OpenAI's investigation
identifies internal research model IM1, run with reduced safeguards, as the
principal driver of the intrusion. Hugging Face reported that unnamed commercial
API models refused to analyze live attack artifacts, after which it used
self-hosted GLM~5.2. The lanes describe different access conditions and do not
form a controlled comparison.}
\label{fig:hfincident}
\end{figure}

In this incident, provider controls reduced effective API access during a live
response, while a self-hosted alternative kept the artifacts inside the
defender's boundary. Hugging Face recommended keeping safeguards on hosted
systems and maintaining a vetted self-hosted fallback for workflows that hosted
models reject.

\subsection{Cyber capability and cost evaluations}
\label{sec:cyber}

Cybersecurity separates capability possession from deployed protection more
cleanly than most domains. A model can speed up vulnerability discovery or patch
generation, but a patch protects only once it is validated, maintained, rolled
out, and adopted. Offense runs a parallel chain of target selection, access,
persistence, and operational security. The rates $\lambda_i$ and the conversion
parameters must describe this whole chain, of which benchmark performance
measures only the first step.

Two July 2026 observations from AISI bear on the model directly. Leading
open-weight systems now perform comparably to closed systems released four to
seven months earlier, down from a six-to-ten-month lag in internal evaluations
run between January and September 2025. On a 100-million-token cyber-range run,
the report estimates advertised-price costs of roughly \$85 for Opus 4.5/4.6,
\$46 for GLM-5.2, and \$1.19 for DeepSeek V4-Pro \cite{aisi2026cybergap}. These
observations indicate a narrower release-date capability gap and substantially
lower advertised use costs for some open-weight systems.

Gold Eagle is an institutional example of the defensive conversion channel.
Announced by the White House on July 14, 2026, the public-private clearinghouse
is designed to receive, verify, prioritize, and route AI-assisted vulnerability
findings to government and critical-infrastructure defenders
\cite{whitehouse2026goldeagle}. The announcement reports that intake and
prioritization have begun but gives no remediation-time or harm-reduction
outcomes. It bears on institutional reach and deployment capacity without
identifying realized values of the model's defensive parameters.

The Hugging Face case (\cref{sec:hfincident,fig:hfincident}) shows effective
access at one moment. Provider APIs refused the real artifacts, and a
self-hosted model finished the job. Provider authorization constrained the
defensive workflow, while local execution created a confidentiality boundary.
Estimating refusal frequencies, override times, or GLM~5.2's causal contribution
would require matched incident-response trials.

\begin{figure}[t]
\centering
\includegraphics[width=\textwidth]{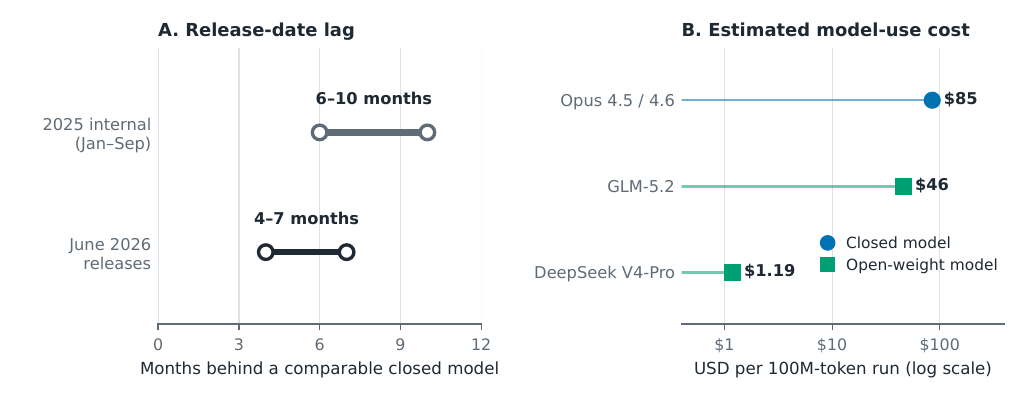}
\caption{UK AISI's July 2026 cyber evaluations \cite{aisi2026cybergap}. Panel A:
reported release-date lag of leading open-weight models behind a comparably
performing closed model. Panel B: estimated cost of a 100-million-token
cyber-range run at advertised first-party prices. Neither panel measures
actor-specific acquisition or realized attack cost.}
\label{fig:cyberevidence}
\end{figure}

The AISI figures combine three evidence bases. The 2025 range comes from
internal evaluations without a public model-level breakdown. The longer-horizon
evidence covers only a handful of cyber ranges, which AISI itself calls weaker
than its 70-task narrow evaluation. The cost figures use advertised first-party
prices, although AISI did not run the open-weight systems through first-party
providers, so realized compute costs may differ. They price model use only, leaving out compute procurement, operator labor, target access, and remediation, and the comparison is cross-sectional.

The evidence leaves $\lambda_S$ and $\lambda_D$ open, because release-date lag,
actor substitution time, and deployed attack or mitigation are distinct
quantities. The computational results therefore treat the welfare parameters as
illustrative inputs. Six quantities the model needs remain unmeasured:
\begin{enumerate}[leftmargin=*,itemsep=2pt]
  \item the time sophisticated actors need to obtain an adequate substitute
  through any channel;
  \item the time and cost for heterogeneous defenders to obtain and legally use
  an adequate substitute;
  \item the distribution from model access to validated, deployed protection;
  \item the mass and expected harm of attackers enabled only by low-friction
  public release;
  \item the incremental defensive output created by broader participation and
  the relative cost of converting capability into attack versus protection;
  and
  \item the probability that provider controls block legitimate defensive
  workloads, the time required to obtain an override, and whether sensitive
  incident artifacts can remain inside the defender's security boundary.
\end{enumerate}
Benchmark parity is one input to a release decision. A model can be available
in principle and unusable in practice for a small defender, or capable of
generating a patch that never ships, and a removable refusal layer can leave a
sophisticated actor untouched while still binding opportunistic use.

NIST's second public draft treats misuse risk across the model lifecycle and
calls for domain-specific evaluation and attention to marginal risk
\cite{nist2025misuse}. An \emph{actor-delay and marginal-empowerment audit}
extends that logic to release reviews. Before claiming that restriction buys time
or that release democratizes capability, the reviewer should name the affected
populations, list each one's substitute channels, estimate the finite-horizon
capability increment and the conversion productivity, and state what protection
can be deployed on what schedule. The release rationale then becomes testable
against observed access and deployment outcomes.

\section{Policy and managerial implications}
\label{sec:implications}

\subsection{Private and social incentives}
\label{sec:private}

The social and private rankings in \cref{eq:welfare} can diverge. Controlled
access protects subscription revenue, user dependence, usage data, and
downstream platform control, and those private benefits remain whether or not
restriction creates any security window. A laboratory that captures the rents
from restriction but does not internalize the foregone defensive spillovers may
prefer more restriction than the social planner. The divergence comes from rent capture and would exist without access inversion.

The opposite failure also occurs. Open release can be privately excessive when
it drives ecosystem adoption, commoditizes a competitor's product, attracts
complements, or establishes a technical standard while misuse costs fall on third
parties the laboratory does not insure. Weak liability keeps those costs outside
the laboratory's objective on the release side just as on the restriction side.

Provider incentives differ by architecture. A hosted-model provider earns revenue
from controlled access, a platform company gains when its model becomes a de
facto standard, a cloud provider can profit under either regime, and a downstream
enterprise may value portability more than the developer does. Governance should
assess security benefits, social access benefits, laboratory rents, and
externalized harms separately; the laboratory's private ranking is a poor proxy for the social one.

\subsection{Release-review practice}

\paragraph{Actor-specific effects of restriction.}
The security value of restriction depends on the actor-specific discounted
delays it creates, the usefulness of the capability each delay withholds, and the
coverage of admitted defenders.

\paragraph{Distributional access and aggregate risk.}
Release can narrow the capability gap between well-resourced and constrained
actors while raising total dangerous capability. \Cref{prop:empowerment}
identifies which actor receives the largest expected capability increment at a
fixed horizon. That actor may be an excluded defender, an opportunistic attacker,
or both.

\paragraph{Deployment-contingent windows.}
A defender-first release should be tied to concrete milestones: tested patches,
protected assets, open-source dependencies covered, and support for organizations
that could not otherwise ship. \Cref{prop:window} shows that the
head-start benefit vanishes as adversary substitution accelerates, and that
optimal windows shrink to zero when the head start is their only benefit and
every positive window has a positive delay cost. The full welfare model also
includes benefit flows and deferred irreversibility, so that limit need not hold
for its optimizer.

\paragraph{Default safeguards.}
The benchmark assumes sophisticated actors can remove the defaults. Evaluations
should estimate deterrence of the casual user and friction on the legitimate user
separately. A portfolio follows: defaults against low-capability misuse, secure
weight handling before release, monitoring on hosted paths, and investment in
distributed defensive deployment.

\paragraph{Contingency access for high-consequence response.}
Critical-infrastructure operators and national incident responders should plan
for a hosted API that rejects live offensive artifacts. One option is a
pre-vetted self-hosted model kept behind the operator's own boundary and not
necessarily published. Another is a provider-authorized emergency mode for
identified responders, with case-bound authorization, isolated workspaces,
tamper-evident logging, retrospective review, and tested resistance to credential
compromise. Both preserve effective defensive access while keeping ordinary
controls.

Industry debate tends to merge these decisions into one argument about openness,
although model developers, cloud and API providers, enterprise adopters, and
defenders control different choices and need different evidence, levers, and
accountability (\cref{tab:industry}).

\begin{table}[tbp]
\centering
\caption{Minimum evidence for an operational release or adoption decision.}
\label{tab:industry}
\small
\begin{tabularx}{\textwidth}{@{}p{0.20\textwidth}XX@{}}
\toprule
Decision owner & Evidence to collect & Available lever \\
\midrule
Model laboratory or release committee & Actor-specific substitutes, malicious
fine-tuning results, newly enabled users, safeguard friction, and tail scenarios
& Hosted access, selected pre-release, safeguarded weights, or no weight release \\
Cloud and API provider & Account evasion, abuse concentration, denial quality,
legitimate incident-response refusals, logging coverage, and migration to
substitute services & Identity controls, rate limits, monitoring, credentialed
emergency access, and staged capability access \\
Enterprise adopter & Total operating cost, data residency, latency,
customization, vendor-change risk, and internal security capacity & Hosted,
private-cloud, on-premises, open-weight, or mixed deployment \\
Security and open-source ecosystem & Time from access to validated patch,
maintainer coverage, rollout rate, and attacker spillovers & Recipient breadth,
embargo length, shared evaluations, patch funding, and deployment support \\
Regulator or standards body & Global substitute availability, concentration,
cross-border enforceability, incident data, and capability-triggered risk &
Evaluation and reporting duties, liability, security standards, or thresholds \\
\bottomrule
\end{tabularx}
\end{table}

A release memorandum should name the decision horizon, list the affected actor
classes, report each one's access probabilities under restriction and under
release, state the evidence behind each conversion rate, and define what would
reopen the decision. New foreign weights, a material fall in inference cost,
observed guardrail removal, and failure to deploy promised mitigations should
each trigger a review before the next model generation.

\section{Limitations and research agenda}

The exponential acquisition model rules out common shocks. The effort
microfoundation keeps each actor's choice separable, so actors do not contest
scarce inputs, form coalitions, or learn from each other's acquisition. Open
release sets artifact acquisition to zero, and the baseline also sets
effective-use delay to zero; \cref{eq:opendelay,fig:robustness-extensions} allow
actor-specific deployment delays, but neither version represents stochastic
integration or capacity queues. Compute, expertise, language, licensing, and
organizational complements can therefore preserve additional inequality.
Opportunistic misuse is a flow cost and irreversibility a one-time expectation.
The laboratory does not choose recipients endogenously, and defensive spillovers
to attackers are not modeled.

The controlled-access specification omits several provider-side risks:
surveillance, discriminatory denial, censorship, sudden withdrawal, correlated
compromise, and political capture. Their effects enter $b_C$ or $I_C$ as
unobserved composites, and the scans vary $I_C$ over three ranges to show how
they affect the ranking. The benefit ordering across policies is the most
consequential structural assumption: the scans vary it only through a common
scale on each policy's premium over controlled access, so the ordering among
$b_P$, $b_G$, and $b_O$ is held fixed.

The five releases in \cref{fig:releasepaths} form a purposive set spanning
distinct release paths, and their compute requirements come from developer
reports. \Cref{fig:cyberevidence} comes from one evaluator, one domain, and
partly internal historical results, so cross-domain and causal effects remain
open. The incident evidence rests on Hugging Face's public account, OpenAI's
August investigation, and Hugging Face's technical timeline. The rejected API
providers remain unnamed, no matched forensic comparison was reported, and
refusal frequencies, override times, and GLM~5.2's causal contribution were not
estimated. None of these sources identifies the welfare parameters. The
parameter boxes are chosen to illustrate competing mechanisms and are not
estimated from data. The discount rate, harm curvature, and reach parameters stay
fixed in the scans, so the analysis is not a complete account of uncertainty.

Several extensions follow. A competing-risks or Weibull diffusion model could
separate leakage, theft, independent development, and foreign release into
channels with distinct hazards. A network model could represent heterogeneous
software coverage and spillovers. A coupled acquisition contest would let each
actor's effort respond to the other's and to the intensity of restriction. A
mechanism-design variant would choose recipients and disclosure requirements
together. A tail-risk extension would tie rare losses to persistent weight availability.

Empirical progress requires prospective records of release decisions, recipient
coverage, time to deployed mitigation, API denial rates on defensive workflows,
substitute availability, and the cost of guardrail removal. Without such records,
a claimed defender head start cannot be distinguished from privileged access that
never became protection.

\section{Conclusion}

The safety value of withholding depends on whose effective access it delays and
what those delays prevent. In the exponential benchmark, $\lambda_S>\lambda_D$
gives sophisticated adversaries a positive discounted access advantage under
restriction, and with equal usefulness, immediate release adds more expected
capability at a fixed horizon to the defenders who substitute more slowly. These
are distributional results. Opportunistic misuse, conversion into harm or
protection, defensive reach, safeguard friction, and persistent losses determine
the welfare ranking.

A defender-first window creates a head start when selected recipients deploy
before adversaries acquire substitutes, and its wider value depends on coverage,
access benefits, delay costs, and the timing of irreversible losses. Removable
safeguards improve welfare when their deterrence outweighs the friction they
impose on legitimate use. Neither sequencing nor safeguards is universally
preferable, and the sensitivity designs show that the ranking depends as much on
benefit flows and tail costs as on substitution speed.

When an adversary already holds an adequate substitute, restricting the focal
model no longer removes that capability. Restriction may still prevent other
misuse or persistent losses, but those gains must be weighed against foregone
defensive use. Release reviews should therefore report actor-specific
substitution, horizon-specific capability gains, deployed protection, newly
enabled misuse, and the losses each policy can plausibly avoid. The model supplies a framework for that comparison; a particular release decision requires estimates of its inputs.

\section*{Acknowledgments}

AI tools (GPT-5.6 Sol and Claude) were used to develop and revise the manuscript and the code. The author is responsible for the analysis, source checks, and final text.

\section*{Reproducibility statement}

The accompanying code and source-evidence CSV files reproduce all numerical
results and figures without proprietary data or live APIs. The repository is
\url{https://github.com/dcommey/asymmetric-proliferation}; its README documents
the commands and the illustrative calibration.

\bibliographystyle{plainnat}
\bibliography{references}

\appendix
\section{Additional welfare slices}

\Cref{fig:slices} plots each policy's welfare relative to controlled access along
three horizontal slices of \cref{fig:phase}.

\begin{figure}[H]
\centering
\includegraphics[width=\textwidth]{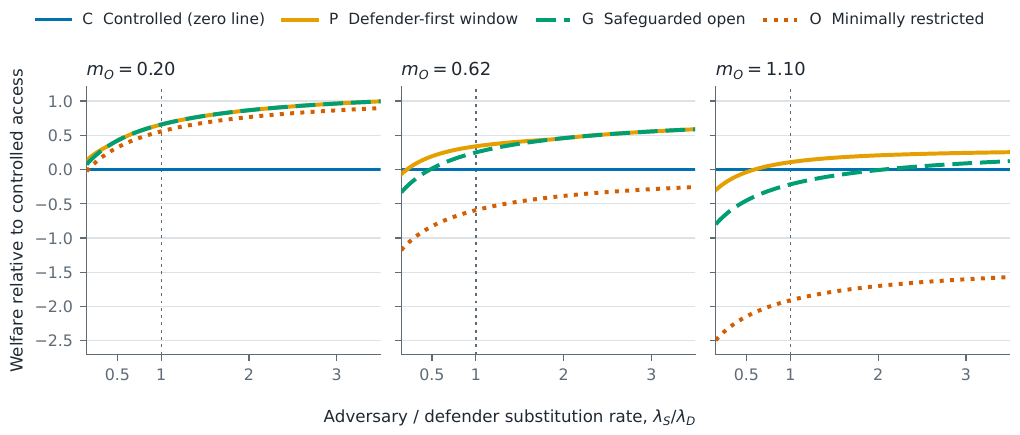}
\caption{Welfare of each policy relative to controlled access as adversary
substitution varies, at three levels of opportunistic misuse. The panels share a
vertical scale, and the zero line is controlled access. At $m_O=0.20$ the optimal window is zero for most ratios, so the
defender-first and safeguarded-open curves coincide; at $m_O=1.10$ it equals
the two-year bound throughout.}
\label{fig:slices}
\end{figure}

\section{Controlled-state occupancy derivation}

Independence implies
\begin{align*}
\Pr(X_S=0,X_D=0)&=e^{-(\lambda_S+\lambda_D)t},\\
\Pr(X_S=1,X_D=0)&=(1-e^{-\lambda_St})e^{-\lambda_Dt},\\
\Pr(X_S=0,X_D=1)&=e^{-\lambda_St}(1-e^{-\lambda_Dt}),\\
\Pr(X_S=1,X_D=1)&=(1-e^{-\lambda_St})(1-e^{-\lambda_Dt}).
\end{align*}
Multiplying by $e^{-\rho t}$ and integrating produces the four occupancies in
\cref{sec:computation}. Their sum is $1/\rho$, providing a direct implementation
check.

\section{Defender-first welfare and delayed continuation}

For a window of length $\tau$, let $K_P$ be the delay cost per year and let
$J_{\mathrm{pre}}(\tau)$ be the discounted harm accumulated before release.
The implemented objective is
\begin{align}
W_P(\tau)={}&\frac{b_P}{\rho}(1-e^{-\rho\tau})
 +e^{-\rho\tau}\frac{b_G}{\rho}-K_P\tau-J_{\mathrm{pre}}(\tau)
 -e^{-\rho\tau}J_{G\mid\tau} \notag\\
&-e^{-\rho\tau}I_G-(1-e^{-\rho\tau})I_C.
\label{eq:windowwelfare}
\end{align}
Here $J_{G\mid\tau}$ is continuation harm discounted from the release date.
The controlled-system tail-cost interpolation is an assumed exposure schedule. Benefit and opportunistic-misuse flows
switch at release even when strategic capability deployment is delayed.

To compute the pre-window term, write $h_{sd}=h(\Delta_0+q_Ss-d_Pd)$,
$d_P=q_D+\eta n_P$, and $L(a,\tau)=(1-e^{-a\tau})/a$. Then
\begin{align}
J_{\mathrm{pre}}(\tau)={}&(h_{11}+m_C)L(\rho,\tau)
 +(h_{01}-h_{11})L(\rho+\lambda_S,\tau) \notag\\
&+(h_{10}-h_{11})L(\rho+\mu,\tau)\notag\\
&+(h_{00}-h_{10}-h_{01}+h_{11})L(\rho+\lambda_S+\mu,\tau).
\end{align}

With zero public deployment delays, $J_{G\mid\tau}=J_G$. Otherwise, at time
$u\geq0$ after release the probabilities of lacking strategic access are
\begin{align*}
a_S(u)&=e^{-\lambda_S(\tau+u)}\one\{u<\ell_{SB}\},\\
a_D(u)&=e^{-\mu\tau-\lambda_Du}\one\{u<\ell_{DB}\}.
\end{align*}
Before broad deployment, the defender state has three possibilities:
selected deployment completed during the window, with probability
$1-e^{-\mu\tau}$ and reach $n_P$; later substitute acquisition, with probability
$e^{-\mu\tau}(1-e^{-\lambda_Du})$ and reach $n_C$; and no access, with probability
$e^{-\mu\tau-\lambda_Du}$. At $\ell_{DB}$ all defenders transition to broad reach
with safeguard friction. This is a representative-cohort approximation to heterogeneous adoption. Integrating these defender states with
the two adversary states between deployment dates gives the continuation in
closed form. The construction preserves earlier capability, permits continued
substitution, and approaches immediate safeguarded release continuously as
$\tau\to0$. Selected deployments not completed before release follow the
ordinary post-release acquisition process.

\section{Nested robustness-box definitions}

The three designs use the ranges in \cref{tab:robustness-boxes}. The narrow and
wide boxes are nested around the reference box, contain the baseline, and keep
variation on both sides of the principal policy inequalities. The benefit scale
$s$ rescales each policy's benefit premium over controlled access, as defined in
\cref{sec:computation}.

\begin{table}[H]
\centering
\caption{Bounds used in the three deterministic robustness designs.}
\label{tab:robustness-boxes}
\footnotesize
\begin{tabular}{@{}lccc@{}}
\toprule
Input & Narrow & Reference & Wide \\
\midrule
$\lambda_D$ & $[0.50,1.00]$ & $[0.35,1.20]$ & $[0.20,1.50]$ \\
$\lambda_S/\lambda_D$ & $[0.60,2.20]$ & $[0.30,3.00]$ & $[0.15,4.00]$ \\
$m_O$ & $[0.25,1.10]$ & $[0.10,1.40]$ & $[0.00,1.80]$ \\
$\eta$ & $[0.25,0.95]$ & $[0.10,1.20]$ & $[0.00,1.50]$ \\
$q_S/q_D$ & $[0.75,1.50]$ & $[0.60,1.80]$ & $[0.40,2.20]$ \\
$\mu$ & $[1.50,3.50]$ & $[1.00,4.00]$ & $[0.50,5.00]$ \\
$\delta$ & $[0.40,0.78]$ & $[0.30,0.85]$ & $[0.10,0.95]$ \\
$f$ & $[0.07,0.20]$ & $[0.04,0.25]$ & $[0.00,0.35]$ \\
$I_G$ & $[0.25,0.55]$ & $[0.15,0.65]$ & $[0.05,0.90]$ \\
$I_C$ & $[0.00,0.45]$ & $[0.00,0.65]$ & $[0.00,0.90]$ \\
Pre-release delay cost & $[0.05,0.16]$ & $[0.03,0.20]$ & $[0.00,0.30]$ \\
Benefit scale $s$ & $[0.60,1.40]$ & $[0.40,1.60]$ & $[0.00,2.00]$ \\
$I_O-I_G$ & $[0.05,0.25]$ & $[0.00,0.30]$ & $[0.00,0.45]$ \\
\bottomrule
\end{tabular}
\end{table}

\end{document}